\documentclass[10pt,twocolumn,a4paper,conference]{IEEEtran}
\usepackage{graphicx} 

\usepackage[latin1]{inputenc} 

\usepackage[T1]{fontenc} 

\usepackage{amsmath,amsfonts,amsbsy,amssymb} 

\usepackage{mathrsfs}

\usepackage[nolist]{acronym} 

\usepackage{tabularx} 

\usepackage{multirow}

\usepackage{wasysym}

\usepackage{float}

\usepackage{color} 

\usepackage{enumitem} 

\usepackage{graphicx,subfigure}

\usepackage{algorithm}
\usepackage{algpseudocode}
\usepackage{amsthm}

\newtheorem{lemma}{Lemma}
\newtheorem{theorem}{Theorem}

\begin{document}
\title{Achieving Ultra-Reliable Communication via CRAN-Enabled Diversity Schemes}
\author{Binod Kharel, Onel L. Alcaraz L\' opez, Hirley Alves, Matti Latva-aho\\
	\IEEEauthorblockA{Centre for Wireless Communications (CWC), University of Oulu, Finland}
	%
	%
	\{binod.kharel, onel.alcarazlopez, hirley.alves, matti.latva-aho\}@oulu.fi
}
%
\maketitle

\begin{abstract}
	Internet of things is in progress to build a smart society, and wireless networks are critical enablers for many of its use cases. In this paper, we present a multi-coordinated transmission scheme to achieve ultra-reliability for critical machine-type wireless communication networks. We take advantage of diversity, which is fundamental for dealing with fading channel impairments, and for achieving ultra-reliable region of operation in order of five 9's as defined by 3GPP standardization bodies. We evaluate an interference-limited network composed of multiple remote radio heads that are allowed to cooperate, by keeping silence thus reducing interference, or by performing more elaborated strategies such as maximal ratio transmission, in order to serve a user equipment with ultra-reliability. We provide extensive numerical analysis and discuss the gains of cooperation by the centralized radio access network.
\end{abstract}


\section{Introduction}

 Fifth-generation mobile networks (5G) will address not only the evolutionary aspects of higher data rates but also the revolutionary aspect of use cases such as massive machine type communication (mMTC) and ultra-reliable low latency communication (URLLC) which poses diverse and stringent requirements \cite{popvoski1}. In the context of URLLC such as factory automation or vehicular communication, which require an extremely high reliability (e.g., frame error rates of $\!10^{-9}\!$ or $\!10^{-5}\!$, respectively) while providing end-to-end delay of 1ms \cite{fettweis1}. Moreover, different requirements face major challenge of reducing latency while providing higher reliability services in the Radio Access Network (RAN) of 5G New Radio (NR), as well as coexistence with different service categories such as enhanced mobile broadband (eMBB). Some of the major challenges and key component issues related to ultra-reliability are  enhanced control channel reliability, link adaptation, interference mitigation. Also, notice that coexistence with other services such as eMBB is related with pre-emptive scheduling and link adaptation \cite{achieving}. These metrics make physical layer design of URLLC very complicated \cite{physical}.

 Multi-connectivity with signal-to-noise ratio (SNR) gain at receiver side using schemes like joint decoding, selection combining and maximal ratio combining (MRC) are discussed in \cite{fettweis2}. Use of various diversity sources, packet design and access protocols as key component supporting URLLC in wireless communication system are discussed in \cite{popvoski2}. Authors in \cite{popvoski3} discussed the way to offer URLLC without intervention in physical layer design by using interface diversity and integrating multiple communication interfaces, where each interface is based on different technology. An energy efficient power allocation strategy for the Chase Combining Hybrid Automatic Repeat Request (CC-HARQ) using finite block-length to achieve ultra-reliable communication are discussed in \cite{shehab1}.
 Authors in \cite{onel1} investigate cooperative communications via relaying protocols to meet ultra-reliable communication (URC) as feasible alternative to typical direct communications framework. In\cite{onel4}, authors focus in the problem of downlink cellular networks with Rayleigh fading and stringent reliability constraint by using topological characteristics of the scenario, and show that  ultra-reliable region is attained by using multiple antennas at the receiving User Equipment (UE) for finite and infinite blocklength coding.
 
Different from \cite{onel4} our work is mainly based on achieving the ultra-reliable region of operation by means of cooperation/coordination of Remote Radio Heads ($\mathrm{RRHs}$) for downlink transmission. Specifically, we consider silencing scheme, where some part of the interfering $\mathrm{RRHs}$ remains in on state and the other are in off, and the Maximum Ratio Transmission (MRT) different from MRC in \cite{onel4} where we allow diversity at transmission side rather at receiving side providing significant diversity gain to cope with very stringent reliability constraints. The multi-connectivity scenario and basis of our system model assumption is mainly based on \cite{fettweis2}. Further, ensuring high reliability using multiple node redundant transmission is also included in the study of enhancements for URLLC support in the 5G Core network in 3GPP (Release 16) \cite{3gpp1}. The system model is CRAN enabled architecture and the main benefit of CRAN architecture is that the signal processing tasks of each small-cell base station (BS) are migrated to Base Band Unit (BBU) pool while enabling coordinated multipoint transmission, centralized resource allocation, joint user scheduling and data flow control \cite{cran1}. The main contribution of this work can be listed as follows: 

\begin{itemize}
	\item we attain accurate closed-form approximations to the outage probabilities when the UE operates under full interference, silencing (mitigating interference by silencing some $\mathrm{RRHs}$) and MRT schemes;
	\item we address the rate control problem constrained on target reliability constraints for the proposed schemes;
	\item numerical results show the superiority of the MRT scheme and the feasibility of ultra-reliable operation when the number of $\mathrm{RRHs}$  increases;
	\item we show that our analytical results are valid by corroborating them via Monte Carlo simulations.
\end{itemize}
Next, Section \ref{sc:system model} introduces the system model and assumptions. Section \ref{sc:reliability diversity} presents the diversity and reliability formulation strategy, while Section \ref{sc:numerical results} shows numerical analysis and rate control under reliability constraints. Finally, Section \ref{sc:conclusion} concludes the paper.

$\textbf{Notation}$: $X\!\sim$$\mathrm{Exp(1)}$ is a normalized exponential random variable with Cumulative Distribution Function (CDF) $F_X(\textit{x})\!\!\!=\!\!\!1-e^{-\textit{x}}$, while $\!Y\!$ is gamma random variable $\Gamma(n,\tfrac{1}{n})$ with the Probability Density Function (PDF) $\!f_Y(\textit{y})\!=\tfrac{\textit{y}^{n-1}\exp^{-\textit{y}}}{(n-1)!}$. Also, $\Gamma(p,x)$=$\int_x^{\infty}t^{p-1}e^{-t}dt$ is the incomplete gamma function, and $_2F_1(a,b;c;z)$ denote the Gaussian regularized hypergeometric function\cite{mathematica}.

\section{System model} \label{sc:system model} 
Consider a multi-node downlink cellular network in which there are $\eta+1$ $\mathrm{RRHs}$ spatially distributed in a given area $\mathcal{A}\subseteq\mathbb{R}^2$. In the topology, there is a typical link ($\mathrm{RRH_0}$) which is assumed to be close to UE and other distributed $\mathrm{RRHs}$ which are equidistant with UE\footnote{Notice that in real word setups the UE could be at any location in a given time. In this work, we have relaxed this by assuming equal distances to interfering nodes for analytical tractability and getting closed-form solutions such that some insights can be discussed.}. The links are further connected to cloud networks where BBU is present by wireless or fixed line connections. The CRAN is enabled with computation and storage units enabling edge computing as shown in Fig.1. We assume all other $\mathrm{RRH}_j,\ j=1,\dots,\eta$ are using the same channel to transmit data to their corresponding user equipment UE. Here, we refer to the typical link between UE and $\mathrm{RRH}_0$ as typical link\footnote{Notice that we focus the analysis on the reference user only.} with length $d_0$ while the  distance between each $\mathrm{RRH}_j $ and UE is denoted by $d_j$, $j=1,\dots,\eta$. We assume channel undergoes quasi-static Rayleigh fading and path loss exponent is denoted by $\alpha$. We focus on the analysis of the typical link's performance when the remaining $\mathrm{RRHs}$ are:
\begin{itemize}
	\item not cooperating (thus, not edge computing or CRAN-enabled).
	\item cooperating through the CRAN\footnote{BBU at CRAN enables coordinated multipoint transmission similar to scenario described in \cite{physical},\cite{fettweis2} and \cite{cran1}.}.
\end{itemize}

\begin{figure}[ht!]
	\centering  
	\subfigure{\includegraphics[width=0.47\textwidth]{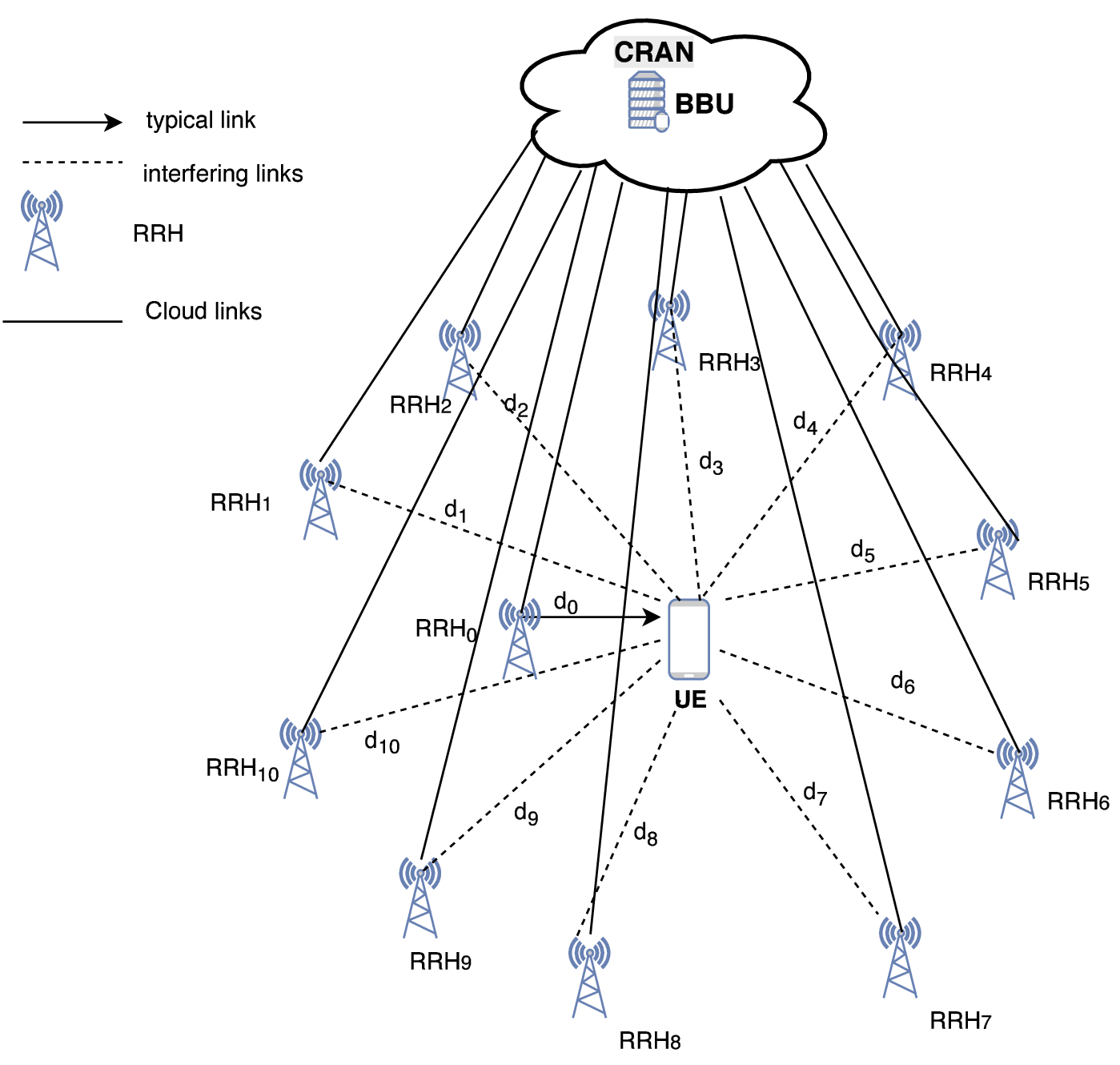}}
	\caption{Illustration of the system model with $\eta=10$.}	
	\label{fig1}
	\vspace*{-2mm}
\end{figure}

Furthermore, UE is equipped with single antenna and we assume that the fading realizations can be treated as independent events and gains from spatial diversity can be fully attained.
 
Consider that each $\mathrm{RRH}$ transmits with fixed unit power and there is a dense network deployment such that system is interference limited, therefore the impact of noise can be neglected. Under these settings, Signal-to-Interference Ratio at UE is given by
\begin{align}\label{eq1}
\mathrm{SIR}=\frac{h_0d_0^{-\alpha}}{\ I_j}
\end{align}
 with, $I_j=\sum_{k+1}^{\eta}h_j d_j^{-\alpha}$ where $k$ is the number of cooperating $\mathrm{RRHs}$ and $I_j$ is the interference from the other remaining $\mathrm{RRHs}$. We denote the squared-envelope coefficients of the typical link and other $\mathrm{RRHs}$ as $h_0$, $h_j\sim\mathrm{Exp}(1)$, respectively. Under these assumptions, we derive the closed-form expression for the outage probabilities under each transmission scheme. The analytical results are corroborated via Monte Carlo simulations and discussed in \ref{sc:numerical results}.

\section{Diversity and Reliability}\label{sc:reliability diversity} 
Herein, we consider the typical link experiences interference from all neighbouring $\mathrm{RRHs}$ due to lack of coordination or of backhaul infrastructure for enabling the CRAN. The CDF of $\mathrm{SIR}$ is $\mathbb{P}(\mathrm{SIR}<\theta)=F_\mathrm{SIR}(\theta)$ and can be formulated for the different schemes in consideration when CRAN performs different strategies to serve UE. Note that threshold is $\theta=2^r-1$ \cite{onel4}, where $r$ is the transmission rate whereas, $\delta=\tfrac{d_0}{d_j}$.

\begin{theorem}
	The CDF of the $\mathrm{SIR}$ when silencing $\mathrm{RRHs}$ by limiting interference at UE side is
	\begin{align}\label{eq2}
	F_\mathrm{SIR}^S(\theta)=1-\Big(1+\delta^{\alpha}\theta\Big)^{-(\eta-k)},
	\end{align}
	with $k=0$ representing the case of full interference.
\end{theorem}

\begin{proof} Please refer to Appendix~\ref{app1}.
\end{proof}

\subsection{Maximum Ratio Transmission (MRT)}
MRT is the scheme where the typical as well as the cooperating RRHs are jointly coordinated in transmission to UE as Channel State Information (CSI) is already assumed to be available at CRAN.
\begin{theorem}
	The CDF of the $\mathrm{SIR}$ in the case when UE is served through MRT is	\vspace{-8.2mm}
\begin{align}\notag\label{eq14}
F_{\mathrm{SIR}}^{MRT}(\theta)&=\frac{1}{\Gamma(\eta-k)}\theta^k\Bigg(1+\delta^{\alpha}\theta\Bigg)^{-\eta}\Gamma(\eta)\\
\notag&\Bigg(\Big(1+\delta^{\alpha}\theta\Big)^{\eta}{_2F_1} (k,\eta;1+k;-\theta)\\
&-{_2F_1} (k,\eta;1+k;\frac{\Big(-1+\delta^{\alpha}\theta\Big)}{\Big(1+\delta^{\alpha}\theta\Big)}\Bigg),
\end{align}	
where $k=0$ models to case full interference scenario and theorem is valid for $\delta<1$.
\end{theorem}

\begin{proof} Please refer to Appendix~\ref{app2}.
\end{proof}

\section{Numerical analysis}\label{sc:numerical results}
Numerical analysis and results are presented in this section to evaluate the system performance in terms of reliability for the proposed schemes. In the analysis, we set $\alpha=3.5,\ \delta=0.5$, unless stated otherwise. The topology consists of $\eta=10\ \mathrm{RRHs}$ located away from the UE of interest such that, $\delta\leq 1$. The $k$ represents the number of cooperating $\mathrm{RRHs}$ out of total number of $\mathrm{RRHs}$ in a given area which are cooperating with the $\mathrm{RRH_0}$ in case of MRT, or in silent mode limiting the interference factor $I_j$ in \eqref{eq1}. Herein, we have used Monte Carlo simulation of $10^7$ runs and some schemes not closely match the targeted reliability of $1-10^{-5}=0.99999$ (five 9's) as it requires longer simulation samples. Further work is required so to improve the accuracy on the tail of the distribution.

\begin{figure}[ht!]
	\centering
	\subfigure{\includegraphics[width=0.47\textwidth]{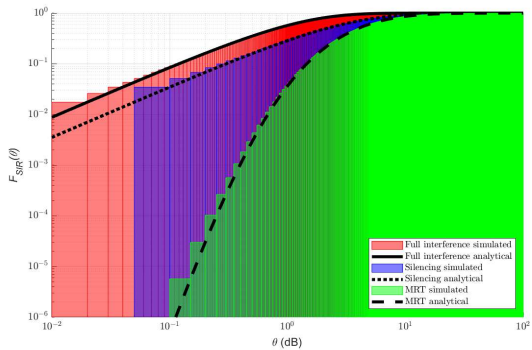}}
	\caption{Overview of SIR distribution for Full interference, Silencing and MRT  schemes with $k=6$ $\mathrm{RRH}$ in coordination and $\eta=10$ where $k=0$ models the case of full interference.}		
	\label{fig2}
	\vspace*{-2mm}
\end{figure}

We compute the CDF of $\mathrm{SIR}$ distribution with respect to threshold $\theta$ (dB) for the proposed schemes as shown in Figure 2. As shown, MRT scheme has left tail distribution already exceeding the value of $10^{-5}$ for same value of $k$ in comparison with silencing scheme. While the left tail of silencing scheme with $k=6 \ \mathrm{RRHs}$ silenced has left tail distribution going greater than $10^{-2}$ in comparison with full interference, but with a reduction in interference factor. The $\mathrm{SIR}$ distribution shows that as value of $k$ increases limiting the interference at UE there is significant improvement at the left tail of the SIR distribution of the analyzed schemes.

 Figure 3 shows the reliability performance of coordinating $k-$out$-\eta$ $\mathrm{RRHs}$ with silencing and MRT schemes. The shaded region in figure represents the ultra-reliable region of operation in case of URLLC which clearly shows that even with $k=4$ $\mathrm{RRHs}$ in cooperation, MRT scheme outperforms all the other schemes. So, the diversity gain from MRT is superior than that of silencing schemes. Although, with more $\mathrm{RRHs}$  in silent mode mitigating interference to UE has significant improvement in reliability. In the figure $k=0$  models the case of full interference. All, the results are validated via Monte Carlo simulations.
 
We generalize the presented topology by comparing all the three schemes in Figure 4 in terms of $\mathrm{SIR}$ threshold $\theta$ (dB). We show that with same value of threshold and distances, MRT scheme easily achieves reliability target of five 9's which is practically infeasible for other schemes. For example, silencing schemes attains reliability target at lower threshold while MRT has some higher threshold for same target reliability. MRT scheme requires prior channel estimation and optimum resource allocation which can be costly at implementation. However, such schemes are being tested in practice \cite{fettweis2}, but not under interference limited constraints considered herein.
\begin{figure}[ht!]
	\centering
	\subfigure{\includegraphics[width=0.47\textwidth]{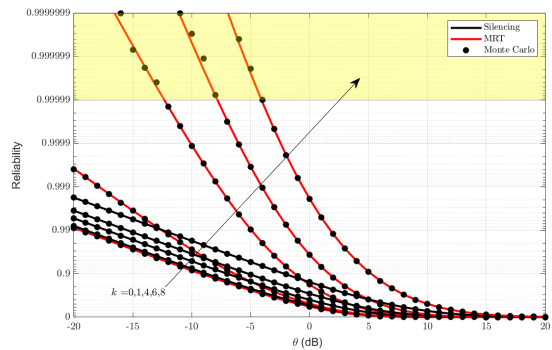}}
	\caption{Overview of Reliability analysis in case for Full interference, Silencing and MRT schemes with different level of coordination of $k$ and $\eta= 10$.}		
	\label{fig3}
	\vspace*{-2mm}
\end{figure}

\subsection{Rate control under reliability constraints}\label{sc:rate_control}

In this section, we evaluate rate control strategies with reliability constraints for the given system model in case of silencing $\mathrm{RRHs}$ and the MRT scheme.
\begin{lemma}
	Assuming silencing scheme, then guaranteeing the reliability constraint given by $\epsilon_{th}$, requires adopting a transmit rate given by
	\begin{align}\label{eq20}
	r=\log_{2}\Bigg(\frac{\epsilon_{th}^{-\frac{1}{(\eta-k)}}-1}{\delta^{\alpha}}+1\Bigg).
	\end{align}
\end{lemma}
In case $k=0$, \eqref{eq20} models the full interference scenario. 
\begin{proof}
	 Note that \eqref{eq20} is final closed-form analytical solution for rate control analysis which comes directly after solving $F_\mathrm{SIR}^S(\theta)=\epsilon_{th}$, for $\theta$ using \eqref{eq2}, where $\theta=2^{r}-1$.
\end{proof}
However, in case for MRT scheme where $F_{\mathrm{SIR}}^{MRT}=\epsilon_{th}$ should be used to calculate the rate from \eqref{eq14}, it is difficult to simplify the equation and invert the term ${_2F_1} (k,\eta;1+k;-\theta)$ and ${_2F_1} \Bigg(k,\eta;1+k;\tfrac{-1+\delta^{\alpha}\theta}{1+\delta^{\alpha}\theta}\Bigg)$ as there is not any standard integral with these hypergeometric functions. In order to evaluate rate analysis we proceed solving numerically
\begin{align}\label{mrt_rate}
\underset{r>0}{\text{arg max}} \   F_{\mathrm {SIR}}^{MRT}(\theta)=\epsilon_{th}.
\end{align}

We evaluate the transmission rate for the given target reliability constraints ($\epsilon_{th}$) for silencing and MRT schemes in Figure 5. In the rate analysis, silencing scheme is evaluated from \eqref{eq20} while MRT scheme is done through numerical analysis from \eqref{mrt_rate}. We show that for given target reliability constraint as number of $k$ increases there is significant improvement over the rate in both cases. MRT allows a significantly higher rate for achieving given target reliability constraint. It means that limiting the interference either from silencing or MRT schemes can enhance the rate for given reliability constraints. Obviously for $k=0$, both the scheme have same rate values.

The cooperation of $\mathrm{RRHs}$  in the vicinity of interference limited network leads to the increase in reliability. However, there is a fundamental question that  should be answered: what is the minimum number of $\mathrm{RRHs}$  needed to cooperate to achieve ultra-reliable communication and have optimum allocation of the resources? For this analysis, we used the problem by formulating the argument based on CDF from \eqref{eq2} and \eqref{eq14} as subjected to case with constrained $\mathrm{Reliability \geq 1-\epsilon_{th}}$.

\begin{figure}[t!] 
	\centering
	\subfigure{\includegraphics[width=0.47\textwidth]{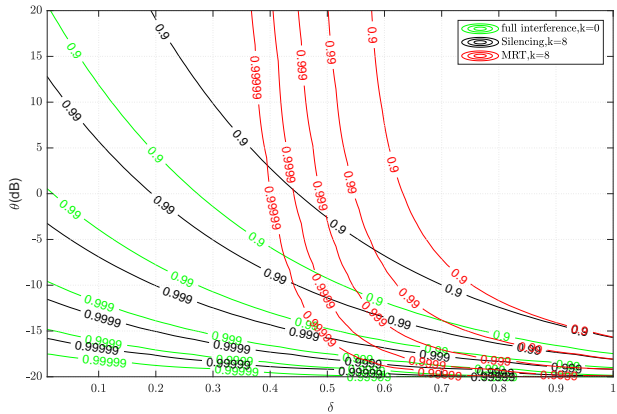}}
	\caption{Illustration of reliability analysis for full interference, Silencing and MRT schemes with respect to $\delta$ and threshold $\theta$ with $\eta=10$ and $k=8$ coordinating $\mathrm{RRHs}$.}
	\vspace*{-2mm}	
	\label{fig4}
\end{figure}

\begin{figure}[t!]
	\centering
	\subfigure{\includegraphics[width=0.47\textwidth]{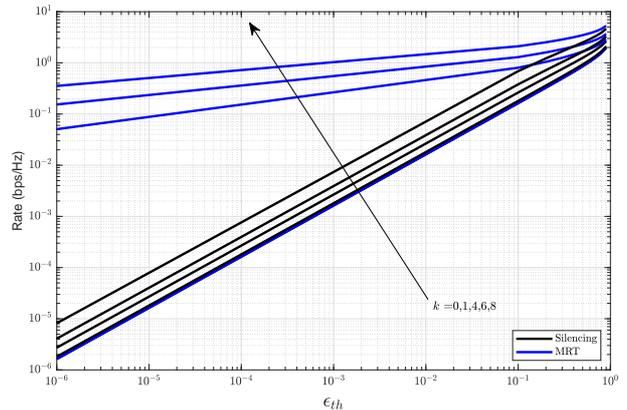}}
	\caption{Illustration of rate control analysis with reliability constraint ($\epsilon_{th}$) in case for Silencing and MRT  schemes with $\eta= 10$.}
	\vspace{-2mm}	
	\label{fig5}
\end{figure}

Figure 6 shows the clear outlook of the optimum number of cooperating $\mathrm{RRHs}$  to meet the given reliability constraints. In the analysis we used the total number of $\mathrm{RRHs}$  in the given area to be $\eta=10$ and we used the threshold $\theta=0.3$ dB. The analysis shows that with MRT scheme even with minimum cooperation of $\mathrm{RRHs}$  can achieve the higher reliability constrains $\epsilon_{th}$ in comparison to silencing schemes. In case of reliability of five 9's which is considered as ultra-reliable operating region only ${k_{min}}=8$ out of 10  in case of MRT can satisfy the target relaibilty but this number can be higher in case of silencing scheme. The increase in threshold $\theta$ can also lead to increase $k$, which requires more resources at CRAN and $\mathrm{RRHs}$.
\vspace{3.5mm}

\begin{figure}[ht!] 
	\centering
	\subfigure{\includegraphics[width=0.47\textwidth]{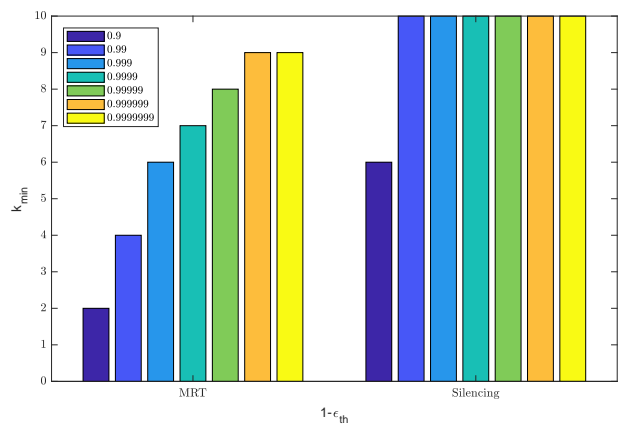}}
	\caption{Overview of minimum required cooperation for Silencing and MRT schemes with the $\eta=10$.}		
	\label{fig6}
	\vspace*{-2mm}	
\end{figure}

\section{conclusion}\label{sc:conclusion}
In this paper, we proposed spatial diversity and multi-connectivity schemes for a downlink cellular system to achive ultra-reliable communication. The performance depends on the transmit rate, distance to the UE, path loss exponent and number of $\mathrm{RRHs}$ in cooperation and interfering. We provide numerical results and discussion by showing outage probability and reliability analysis of the schemes when varying different system parameters. The numerical results show the performance of MRT scheme and also the feasibility of ultra-reliable operations when the number of cooperating $\mathrm{RRHs}$  increases. In case of moderate reliability, silencing scheme is also  feasible. We reached accurate closed-form solution for the $\mathrm{SIR}$ distribution for full interference, silencing and MRT schemes. Finally, our analytical closed-form results are corroborated via Monte Carlo solutions.

\section*{Acknowledgement}
This research has been financially supported by Academy of Finland 6Genesis Flagship (grant 318927), Academy of Finland (Aka) (Grants n.319008, n.307492) and by the Finnish Funding Agency for Technology and Innovation (Tekes), Bittium Wireless, Keysight Technologies Finland, Kyynel, MediaTek Wireless, Nokia Solution and Networks.

\appendices\vspace{-2mm}
\section{Proof of Theorem 1}\label{app1}
In case of silencing some RRHs, interference reduces to $I_j=\sum\limits_{j=k+1}^{\eta}h_jd_j^{-\alpha}$ then we proceed as follows from \eqref{eq1} 		   
\begin{align}\notag\label{eq3}
F_\mathrm{SIR}^S(\theta|h_j)&=\mathbb{P}\Big(\mathrm{SIR}<\theta\Big)=\mathbb{P}\Big(h_0<\theta\delta^{\alpha}\sum\limits_{j=k+1}^{\eta}h_j\Big)\\  &=1-\exp\Big(-\theta\delta^{\alpha}\sum\limits_{j=k+1}^{\eta}h_j\Big)
\end{align} 
where the last step comes from using the CDF of $h_0$. Next, we decondition \eqref{eq3} of $h_j$ using the fact that $\mathrm{Q}=\sum\limits_{j=k+1}^{\eta}h_j$ follows a  gamma distribution \cite{gamma1} with PDF,  $\mathrm{f_Q}(\mathrm{q})=\frac{\mathrm{q}^{(\eta -k-1)}\mathrm{e}^{-\mathrm{q}}}{(\eta -1)!}$. Thus we proceed to calculate CDF as
\begin{align}\notag\label{eq5}
F_\mathrm{SIR}^S(\theta)&=\int\limits_{0}^{\infty}\ F_\mathrm{SIR}(\theta|h_j)\mathrm{f_Q}(\mathrm{q})\mathrm{dq}\\
\notag&\stackrel{(a)}{=}\!1\!-\!\int\limits_{\!0\!}^{\!\infty\!}\!\Bigg(\!\!\exp\bigg(\!-\mathrm{q}\delta^{\alpha}\theta\!\bigg)\!\Bigg)\frac{\mathrm{q}^{\eta-k-1}}{(\!\eta-k-1\!)!}\mathrm{dq}\\
\notag&\stackrel{(b)}{=}\!1\!-\Bigg(\frac{1}{{\!1+\delta^{\alpha}\theta}}\!\Bigg)^{\eta-k}\frac{1}{(\!\eta\! -\!1\!)!}\int\limits_{0}^{\infty}\mathrm{e}^{\!-x\!}x^{\!\eta-1\!}dx\\
&\stackrel{(c)}{=}1-\Big(1+\delta^{\alpha}\theta\Big)^{-(\eta-k)}	
\end{align}
where (a) comes from substituting the respective CDFs, the integral in (b) is the definition of gamma function \cite{gamma1} and reduces $\Gamma(\eta)=(\eta-1)!$ since $\eta\in\mathbb{N^*}$, which simplifying renders (c), concluding the proof.

\section{Proof of Theorem 2}\label{app2}
 Since $\mathrm{RRH_0}$ cooperates with $k \ \mathrm{RRHs}$ the $\mathrm{SIR}$ in \eqref{eq1} becomes
\begin{align}\label{sir_mrt}
\mathrm{SIR}^{MRT}=\frac{h_0d_0^{-\alpha}+\sum\limits_{i=1}^{k}h_id^{-\alpha}}{\sum\limits_{j=k+1}^{\eta}h_jd^{-\alpha}},
\end{align} whose CDF is
\begin{align}\notag\label{eq15}
&F_{\mathrm{SIR}}^{MRT}(\theta)=\mathbb P\Bigg({\mathrm{SIR}^{MRT}}<\theta\Bigg)\stackrel{(a)}{=}\mathbb{P}\Bigg(h_0<\delta^{\alpha}(\theta\mathrm{q-p})\Bigg)\\
&\stackrel{(b)}{=}\!\!\int\limits_{\frac{\!\!\mathrm{p}}{2^r-1}\!\!}^{\infty}\!\int\limits_{0}^{\infty}\!\!\Bigg(\!1\!-\!\exp\!\Big(\!\!{-\delta^{\alpha}\!\big(\theta\mathrm{p-q}\big)\!\Big)\!}\!\Bigg)\!\mathrm{f_P}(\mathrm{p})\mathrm{f_Q}(\mathrm{q})\ \mathrm{dp}\ \mathrm{dq},
\end{align}
where (a) asssumes that $d=d_j$, $\mathrm{P}=\sum\limits_{i=1}^{k}h_i\sim\Gamma (k,\tfrac{1}{k})$ and
$\mathrm{Q}\sim\Gamma(\eta-k,\tfrac{1}{\eta-k})$, the integral in (b) comes after applying the CDFs of $h_0$ and de-conditioning on $\mathrm{p}$ and $\mathrm{q}$. From definition of hypergeometric function \cite{mathematica} we obtain \eqref{eq14}. Note that \eqref{eq14} is our final closed-form solution for $\delta<1$, while $\delta>1$ is neglected since the UE connects to the closest $\mathrm{RRH}$.  

\bibliographystyle{IEEEtran} 
\bibliography{refs}

\end{document}